\documentclass[letterpaper]{article}
\usepackage[utf8]{inputenc}
\usepackage[T1]{fontenc}
\usepackage{lmodern}

\usepackage{graphicx}
\usepackage{amsmath,amssymb,amsthm}
\usepackage[english]{babel}
\usepackage{hyperref}
\usepackage{authblk}
\usepackage{comment}

\makeatletter

\numberwithin{equation}{section}
\numberwithin{figure}{section}
\theoremstyle{plain}
\newtheorem{thm}{\protect\theoremname}
  \theoremstyle{plain}
  
  \theoremstyle{definition}
  
  \theoremstyle{plain}
  \newtheorem{lem}[thm]{\protect\lemmaname}

\usepackage[noend]{algpseudocode}
\usepackage{booktabs}

\makeatother

  \providecommand{\corollaryname}{Corollary}
  \providecommand{\definitionname}{Definition}
  \providecommand{\lemmaname}{Lemma}
\providecommand{\theoremname}{Theorem}

\begin{document}

\title{An upper bound on the $k$-modem illumination problem.}

\newcommand*\samethanks[1][\value{footnote}]{\footnotemark[#1]}

\author{Frank Duque, Carlos Hidalgo-Toscano}
\author{Frank Duque \thanks{Departamento de Matemáticas, Cinvestav, D.F. México, 
México. Partially supported by grant 153984 (CONACyT, Mexico).
\texttt{[frduque, cmhidalgo]@math.cinvestav.mx}} \protect\quad Carlos Hidalgo-Toscano 
\samethanks}

\maketitle

\begin{abstract}
A variation on the classical polygon illumination problem was introduced in 
[Aichholzer et. 
al. EuroCG'09]. In this variant light sources are replaced by wireless devices called $k$-modems, 
which can penetrate a fixed number $k$, of ``walls''. A point in the interior of a polygon is 
``illuminated'' by a $k$-modem if the line segment joining them intersects at most $k$ edges of the 
polygon. It is easy to construct polygons of $n$ vertices where the number of $k$-modems required 
to illuminate all interior points is $\Omega(n/k)$. However, no non-trivial upper bound is known. 
In this paper we prove that the number of $k$-modems required to illuminate any polygon of $n$ 
vertices is at most $O(n/k)$. For the cases of illuminating an orthogonal polygon or a set of 
disjoint orthogonal segments, we give a 
tighter bound of $6n/k+1$.  Moreover, we present an $O(n \log n)$ time algorithm to achieve this 
bound. 
\end{abstract}

\section{Introduction}

The classical art gallery illumination problem consists on finding the minimum number of light 
sources needed to illuminate a simple polygon. 
There exist several variations on this problem; one 
such variation was introduced in \cite{monotone}, it is known as the $k$-modem illumination 
problem. 
For a non-negative number $\ensuremath{k}$, a $\emph{\ensuremath{k}-modem}$ is a wireless device  
that can penetrate $\ensuremath{k}$ ``walls''. Let $\ensuremath{\mathcal{L}}$ be a set of 
$\ensuremath{n}$ line segments (or lines) in the plane. A $\ensuremath{k}$-modem 
$\emph{illuminates}$ all points $\ensuremath{p}$ of the plane such that the interior of the line 
segment joining $\ensuremath{p}$ and the $\ensuremath{k}$-modem intersects at most $\ensuremath{k}$
elements of $\ensuremath{\mathcal{L}}$. In general, $\ensuremath{k}$-modem illumination problems 
consist on finding the minimum number of $\ensuremath{k}$-modems necessary to illuminate a certain 
subset of the plane, for a given $\ensuremath{\mathcal{L}}$. Classical illumination 
\cite{sack1999handbook,o1987art} is just the case when $\ensuremath{k=0}$.


Several upper bounds have been obtained for various classes of $\mathcal{L}$. In \cite{monotone} 
the authors studied the case when $\ensuremath{\mathcal{L}}$ is the set of edges of a monotone 
polygon with $n$ vertices; they showed that the interior of the polygon can be illuminated with 
at most 
$\left \lceil\frac{n}{2k}\right \rceil$ $k$-modems ($\left \lceil\frac{n}{k+4}\right \rceil$ if 
$k=1,2,3)$, and if the polygon 
is orthogonal it can be illuminated with $\left \lceil \frac{n-2}{2k+4}\right \rceil $ 
$k$-modems. In 
\cite{ktransmiters} the authors studied the case when $\mathcal{L}$ is a set of $n$ disjoint 
orthogonal line segments; they showed that $\left \lceil 
\frac{n+1}{2(k+1)^{0.264}}\right \rceil $ 
$k$-modems are sufficient to illuminate the plane. 
In \cite{fabila2009modem} the authors studied the problem of illuminating the plane with few modems 
of 
high power; they showed that when $\mathcal{L}$ is an arrangement of lines, one 
$\left \lceil\frac{3n}{2}\right \rceil$-modem is sufficient, when $\mathcal{L}$ is the set of edges 
of an 
orthogonal polygon, one $\left \lceil\frac{n}{3}\right \rceil$-modem is sufficient and when 
$\mathcal{L}$ is the 
set of edges of a simple polygon, one $\left \lceil\frac{2n+1}{3}\right \rceil$-modem is 
sufficient. 
It is worth 
noting that there are no published bounds for general polygons.

There are also algorithmic results regarding simple polygons. In \cite{metaheuristic} the authors 
presented a hybrid metaheuristic strategy to find few $k$-modems that illuminate a simple polygon. 
In 
this case the $k$-modems are required to be placed at vertices of the polygon. They applied the 
hybrid metaheuristic to random sets of simple, monotone, orthogonal and grid monotone orthogonal 
polygons. Each set consisted of 40 polygons of 30, 50, 70, 100, 110, 130, 150 and 200 vertices.
The average numbers of $k$-modems used by their strategy are shown in Table \ref{table:averages}.

\begin{table}[h]
\centering
\begin{tabular}{ c c c c c }
	\toprule
	 & Simple & Monotone & Orthogonal & Grid Monotone Orthogonal \\
	\midrule
	$k=2$ & $\left \lceil \frac{n}{26} \right \rceil$ & $\left \lceil \frac{n}{15} \right 
\rceil$ & 
$\left \lceil \frac{n}{27} \right \rceil$ & $\left \lceil \frac{n}{18} \right \rceil$ \\
& & & & \\
	$k=4$ & $\left \lceil \frac{n}{52} \right \rceil$ & $\left \lceil \frac{n}{26} \right 
\rceil$ & 
$\left \lceil \frac{n}{57} \right \rceil$ & $\left \lceil \frac{n}{35} \right \rceil$ \\
	\bottomrule
\end{tabular}

\begin{centering}
\protect\caption{Averages on the number of $k$-modems obtained by the strategy presented in 
\cite{metaheuristic}. }
\label{table:averages}
\par\end{centering}

\end{table}

It is known that the problem of finding the minimum number of 0-modems to illuminate a simple 
polygon is NP-hard \cite{Aggarwal:1984:AGT:911725}. It was recently proved that the same problem is 
also NP-hard for $k$-modems {[}reference to the paper of Christiane{]}.

This paper is organized as follows. In Section \ref{SeccKmodemCutting} we present a new bound of
$O(n/k)$ for the $k$-modem illumination problem through a variation of the Cutting Lemma.
In Section \ref{sec:Orthogonal} we present a simple $O(n\log n)$ time algorithm that illuminates a 
set of $n$ disjoint orthogonal segments with at most $6\frac{n}{k}+1$ $k$-modems. This 
algorithm can be easily modified to obtain the same bound for orthogonal polygons. 

\section{\label{SeccKmodemCutting}Upper Bounds Using the Cutting Lemma}

A generalized triangle is the intersection of three half planes. Note that a generalized triangle 
can be a point, a line, a bounded or an unbounded region. Given a set $\mathcal{L}$ of $n$ lines in 
the plane and $r>0$, a $\frac{1}{r} $-cutting is a partition of the plane into 
generalized triangles with 
disjoint interiors, such that each generalized triangle is intersected by at most $n/r$ lines. The 
Cutting Lemma gives an upper bound on the size of such a cutting.

\begin{thm}[Cutting Lemma]
 Let $\mathcal{L}$ be a set of $n$ lines in the plane and $r>0$. Then there exists a 
$\frac{1}{r}$-cutting of 
size $O(r^{2})$. 
\end{thm}

The Cutting Lemma was first proved in \cite{cuttingLemma1} and independently in 
\cite{CuttingLemma2}. It has become a classical tool in Computational Geometry used mainly in 
divide-and-conquer algorithms. It can be used to obtain an upper bound on the number of $k$-modems 
necessary to illuminate the plane in the presence of $n$ lines.

\begin{thm}\label{thm:Lines-Mdems-Required}
 Let $\mathcal{L}$ be a set   of $n$ lines in the plane. The number   of $k$-modems required to 
 illuminate the plane in the presence of $\mathcal{L}$ is $O(n^2/k^2)$.
\end{thm}

\begin{proof}
  By the Cutting Lemma, there exists a $\frac{1}{(n/k)}$-cutting for $\mathcal{L}$ of size 
$O(n^2/k^2)$. Note that each triangle can be illuminated with one $k$-modem. Thus the number of 
$k$-modems required to illuminate the plane is $O(n^2/k^2)$.
\end{proof}

Given a polygon $\mathcal{P}$ we can obtain the same 
bound on the number of $k$-modems needed to illuminate it. We first  extend its edges to straight 
lines and then apply Theorem \ref{thm:Lines-Mdems-Required}. 
We can achieve a  better upper bound  of $O(n/k)$ by making use of a line segment version of the 
Cutting 
Lemma given in \cite{berg1995cuttings}. In that paper the authors consider cuttings in a more 
general setting: they define a cutting as a subdivision of the plane in \emph{boxes}. A \emph{box} 
is a closed subset of the plane which has constant description (that is, it can be represented in a 
computer with $O(1)$ space, and it can be checked in constant time whether a point lies in a box or 
whether an object intersects (the interior of) a box).

\begin{thm} \label{thm:CuttingLemmaSegments} \cite{berg1995cuttings}
Let $\mathcal{L}$ be a set of $n$ line segments in the plane with a total of $A$ intersections and 
$r>0$. Then there exists a $\frac{1}{r}$-cutting for $\mathcal{L}$ of size $O\left(r +A 
\left(\frac{r}{n}\right)^2\right)$.
\end{thm}

If we consider a polygon as a set of $n$ line segments that intersect only at their endpoints, 
Theorem 
\ref{thm:CuttingLemmaSegments} gives us a $\frac{1}{(n/k)}$-cutting of size 
$O\left(\frac{n}{k}+n\left(\frac{n^2}{k^2}\frac{1}{n^2}\right)\right)=O\left(\frac{n}{k}\right)$. 
Taking into account that the boxes 
used in the proof of Theorem \ref{thm:CuttingLemmaSegments} are generalized trapezoids (that is, 
intersections of four half planes), and that it is possible to illuminate each trapezoid with a 
$k$-modem, we obtain Theorem \ref{thm:CuttingLemmaPolygons}. Note that the same reasoning applies to
sets of segments in the plane, which gives a bound of $O(n/k +A/k^2)$ for that case.

\begin{thm}\label{thm:CuttingLemmaPolygons}
Let $\mathcal{P}$ be a polygon with $n$ vertices. The number of $k$-modems needed to illuminate 
$\mathcal{P}$ is at most $O\left(n/k\right)$.
\end{thm}

\section{\label{sec:Orthogonal}An Algorithm for Orthogonal Line Segments Illumination}

In this section we present an $O(n \log n)$ time algorithm  to illuminate the plane 
with
$k$-modems in the presence of a set $\mathcal{L}$ of $n$  disjoint orthogonal segments. The 
number of $k$-modems used by our algorithm is at most $6\frac{n}{k}+1$. 

We assume that $\mathcal{L}$ is contained in a rectangle $R$. Our objective is to partition $R$ 
into a certain kind of polygons called staircases, in a similar fashion to the cuttings introduced 
in Section 
\ref{SeccKmodemCutting}. We do this in such a way so that each staircase is illuminable with one 
$k$-modem. 
A \emph{staircase} is an orthogonal polygon $P$ such that: $P$ is bounded from below by a single 
horizontal segment \emph{Floor(P)}; $P$ is bounded from the right by a single 
vertical segment \emph{Rise(P)}; the left endpoint of $Floor(P)$ and the upper endpoint of 
$Rise(P)$ are joined by a monotone polygonal chain \emph{Steps(P)}. See Figure 
\ref{fig:staircase}. In what follows 
let $P$ be a 
staircase.

\begin{figure}[ht]
	\centering
	\includegraphics[scale=.34]{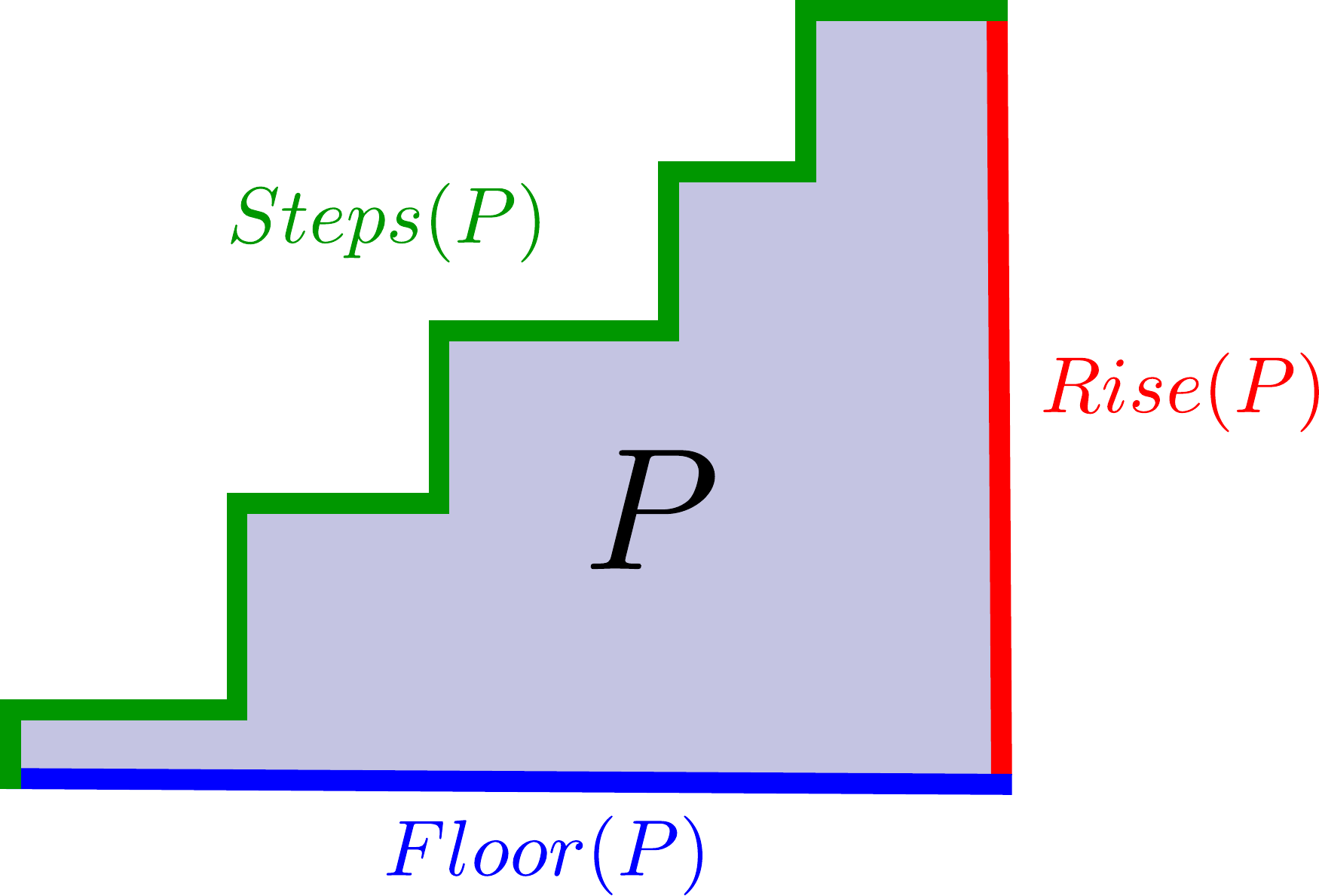}
	\caption{An example of a staircase.}
	\label{fig:staircase}
\end{figure}

Note that an axis parallel line intersecting the interior of $P$ splits it into two parts which are 
also staircases. Let $l$ be a horizontal line that intersects the interior of $P$ in a segment $s$.
We denote by $Above(P,l)$ the staircase formed by $s$ and the part of $P$ above $s$, and denote 
by $Below(P,l)$ the staircase formed by $s$ and the part of $P$ below $s$.
Likewise, let $l'$ be a vertical line that intersects the interior of $P$ in a segment $s$.
We denote by $Left(P,l')$ the staircase formed by $s$ and the part of $P$ to the left of $s$, and 
denote by
$Right(P,l')$ the 
staircase formed by $s$ and the part of $P$ to the right of $s$. See Figure \ref{fig:splits}.

\begin{figure}[ht] 
	\centering
	\includegraphics[scale=.55]{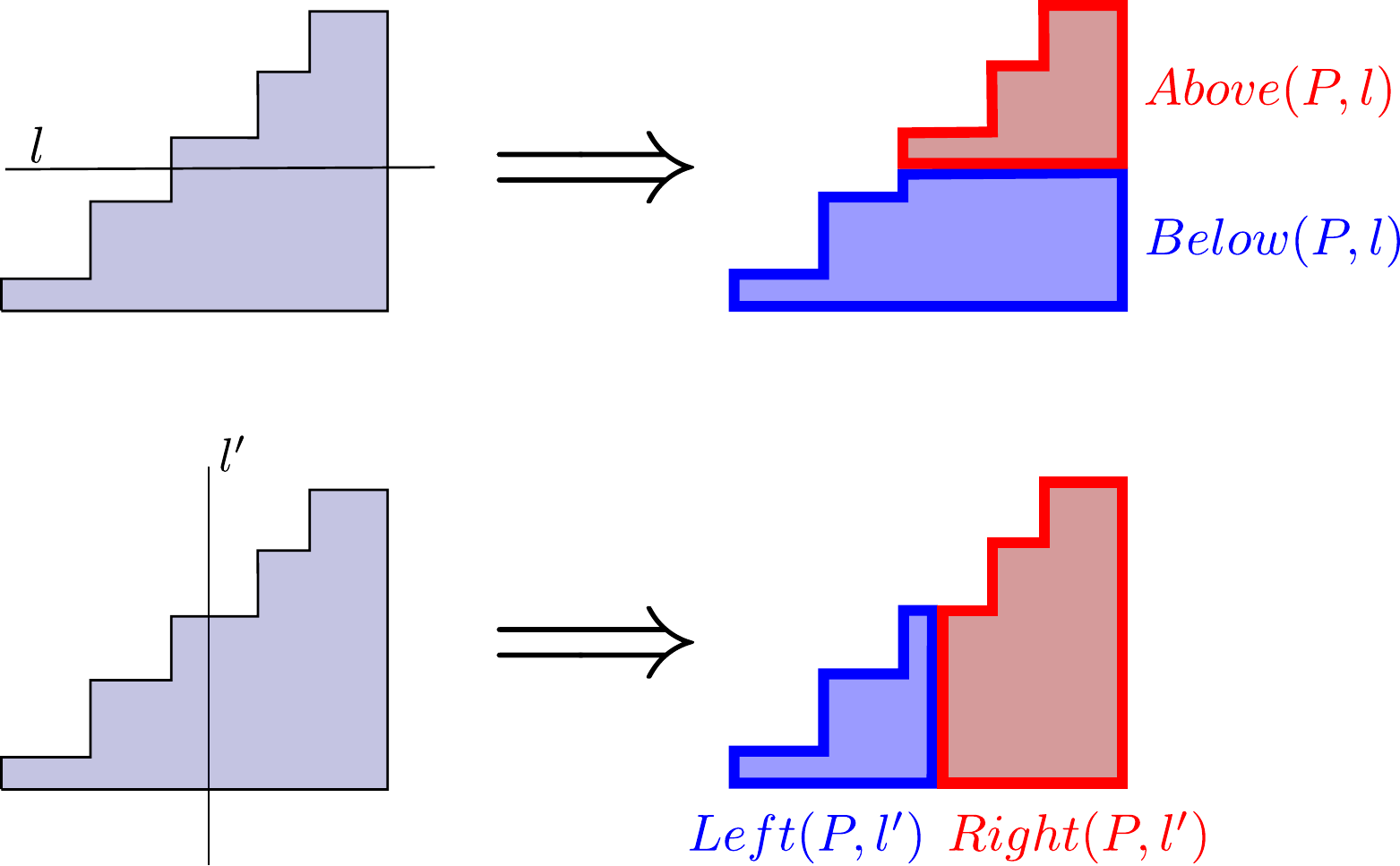}
	\caption{An example of $Above(P,l)$, $Below(P,l)$, $Left(P,l')$ and $Right(P,l')$.}
	\label{fig:splits}
\end{figure}

\subsection{\label{sub:Partition-scheme-pseudocode}Illumination Algorithm}

Let $\mathcal{P}$ be the set of endpoints of the segments in $\mathcal{L}$, excepting the 
rightmost points of the horizontal segments. The algorithm 
\emph{StaircasePartition($\mathcal{L})$} finds a rectangle $R$ that contains $\mathcal{L}$ 
and produces as output a partition of $R$ into staircases.
Initially, $R$ is the only staircase in the 
partition, we use a horizontal sweep line $l$ starting at the top of $R$ that stops at each point 
of 
$\mathcal{P}$ and checks whether it is necessary to refine the partition.

Throughout the algorithm, the staircases that are not intersected by the sweep 
line need no further processing.  
The staircases that are intersected by the sweep line might be splitted later. 
There are two reasons to split a staircase: the 
number of segments of $\mathcal{L}$ that intersect it might be too high, or a horizontal segment of 
$\mathcal{L}$ might cross the staircase completely. These two cases are handled by procedures 
called \emph{OverflowCut} and \emph{CrossingCut} respectively. \emph{OverflowCut} ensures that no 
staircase is intersected by more than $k$ segments; \emph{CrossingCut} ensures that no horizontal 
segment in $\mathcal{L}$ crosses completely a staircase.
 
\emph{OverflowCut} is called whenever $k$ segments intersect the interior of a staircase $P$. 
\emph{OverflowCut} starts by replacing $P$ by $Above(P,l)$. If at most $\left \lfloor 
\frac{k}{2}\right \rfloor$ segments intersect both $P$ and $l$, 
$Below(P,l)$ is added to the partition. If more than $\left \lfloor 
\frac{k}{2}\right \rfloor$ segments intersect both $P$ and $l$, we search for a vertical line $l'$ 
that 
leaves at 
most $\left \lfloor 
\frac{k}{2}\right \rfloor$ of those segments to each side. The staircases 
$Left(Below(P,l), l')$ and $Right(Below(P,l),l')$ are added to the partition. See Figure 
\ref{fig:Overflow}.

\begin{figure}[ht] 
	\centering
	\includegraphics[scale=.45]{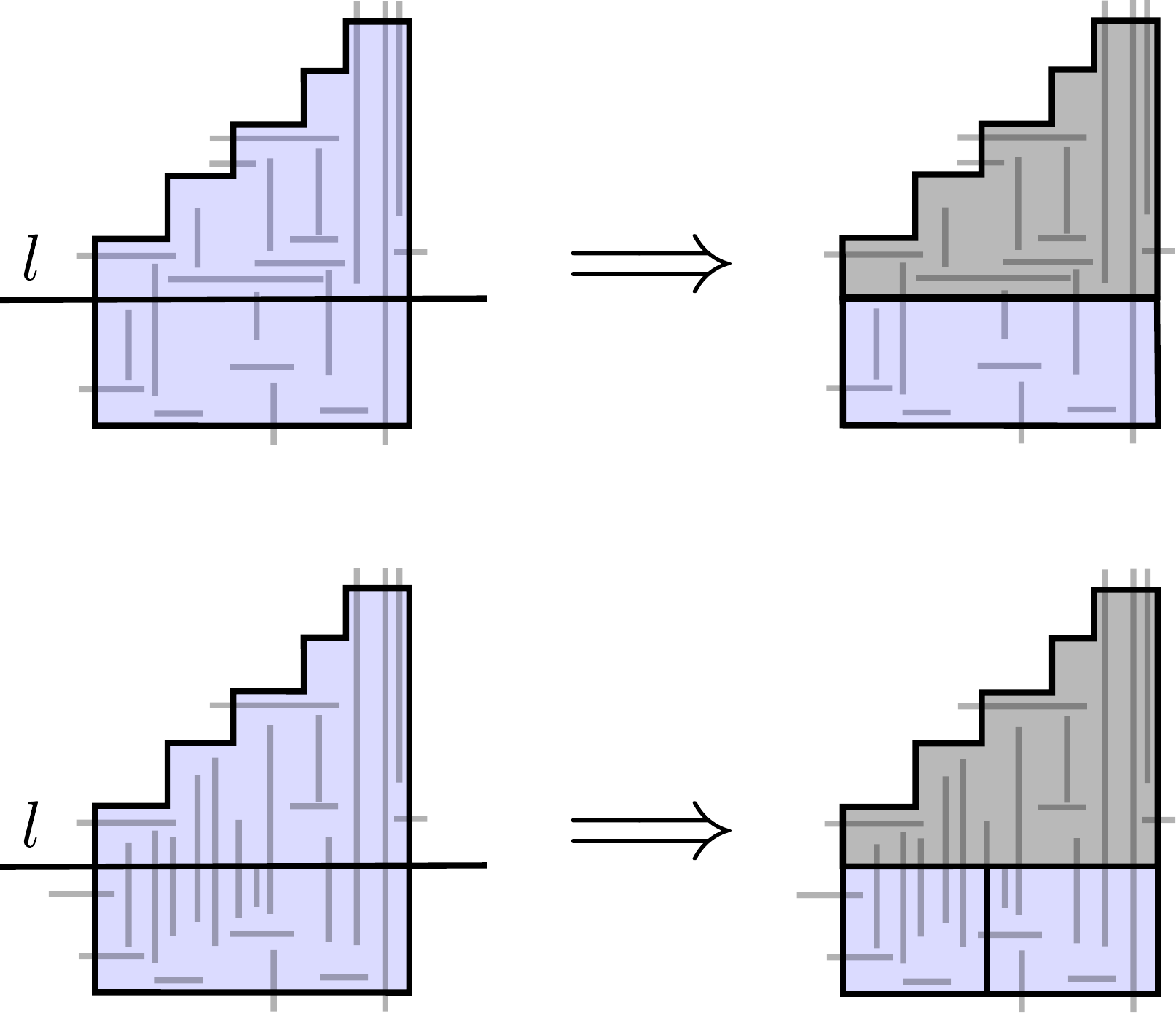}
	\caption{The two cases for an OverflowCut on a staircase with $k=18$.}
	\label{fig:Overflow}
\end{figure}

\emph{CrossingCut} is called whenever a horizontal segment $s$ of $\mathcal{L}$ intersects at least 
three staircases of the partition. Let $P_1, \ldots , P_m$ be these staircases ordered from left to 
right. \emph{CrossingCut} replaces each $P_i$ by $Above(P_i, l)$ for $2 \leq i \leq  m-1$, and 
merges $Below(P_i,l)$ with $P_m$. Note that the number of staircases in the partition does not 
grow. See Figure \ref{fig:Crossing}.

\begin{figure}[ht] 
	\centering
	\includegraphics[scale=.35]{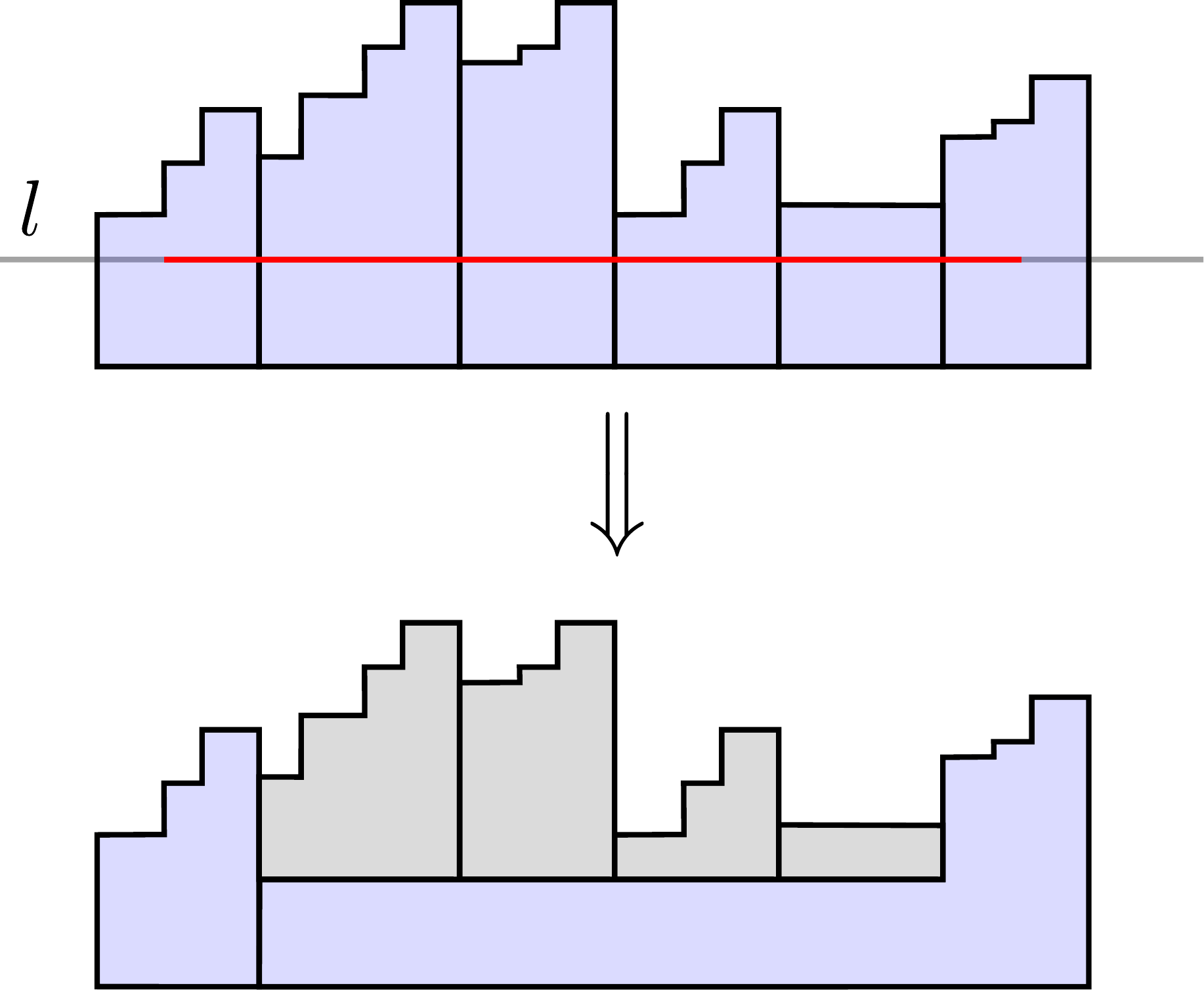}
	\caption{A CrossingCut.}
	\label{fig:Crossing}
\end{figure}

The algorithm begins by sorting the points in $\mathcal{P}$ by $y$ coordinate, which will be the 
stop points for the sweep line. After initializing the partition with $R$, the stop 
points are processed according to three cases. (As part of the steps taken in each case, we 
maintain the 
set of segments that intersect each staircase in the partition sorted from left to right, making a 
distinction of the 
ones that also intersect $l$.) 

\begin{itemize}

\item The sweep line stops at an upper endpoint. This means we have found a new vertical segment 
that 
intersects a staircase $P$. That staircase may now be intersected by $k$ segments, if this happens 
we make an \emph{OverflowCut}.

\item The sweep line stops at a lower endpoint. This means we have 
found the end of a segment that intersects $P$. In this case we only update the set of $P$.

\item The sweep line stops at a left endpoint $p$ of a segment $s$. If the number 
of staircases that 
$s$ intersects is at least three, we perform a \emph{CrossingCut} on them. 
If necessary, an \emph{OverflowCut} is made on the staircases that contain the 
endpoints of $s$.

\end{itemize}

At the end of the algorithm, the interior of each staircase $P$ in the partition 
is intersected by at most $k$ segments. Thus, a $k$-modem 
placed in the intersection between \emph{Floor(P)} and \emph{Rise(P)} is enough to illuminate it.

\begin{lem} \label{lem:num-modems}
  Let $\mathcal{L}$ be a set of $n$ disjoint orthogonal segments. Let $\mathcal{P}'$ be the set of 
endpoints of the segments in $\mathcal{L}$, excepting the 
lower endpoints of the vertical segments. Then, the total number of calls 
to \textbf{OverflowCut} made by \textbf{StaircasePartition($\mathcal{L}$)} is at most 
$|\mathcal{P}'|/ \left \lceil k/2 \right \rceil$.
\end{lem}
\begin{proof}
Let $P$ be a staircase in the partition that \emph{StaircasePartition($\mathcal{L})$} returns. 
$P$ was created after executing an \emph{OverflowCut}, so $Steps(P)$ was initially intersected by 
at most 
$\left \lfloor k/2 \right \rfloor$ vertical segments. Further modifications to $P$ are done 
by \emph{CrossingCuts}, which add edges to $Steps(P)$. The 
horizontal edges added to $Steps(P)$ are part of segments in $\mathcal{L}$, so they can't be 
intersected by any other segment. Thus, at most $\left \lfloor k/2 \right \rfloor$ vertical 
segments of 
$\mathcal{L}$ cross completely $P$.


If \emph{OverflowCut} is called on $P$, $Above(P,l)$ becomes a staircase of 
the partition and won't be modified anymore. \emph{OverflowCut} was called because $k$ segments 
from $\mathcal{L}$ intersected $P$. These same segments intersect also
$Above(P,l)$; of these at most $\left \lfloor k/2 \right \rfloor$ cross it completely. Therefore, 
at least $\left \lceil k/2 
\right \rceil$ points of $\mathcal{P}'$ are contained in $Above(P,l)$. Thus, the total number of 
\emph{OverflowCut} calls is at most $|\mathcal{P}'|/ \left \lceil k/2 \right \rceil$.

\end{proof}

\begin{thm}
 Let $\mathcal{L}$ be a set of $n$ orthogonal disjoint segments. Then the number of $k$-modems 
needed to illuminate the plane in the presence of $\mathcal{L}$ is at most $6\frac{n}{k}+1$. The 
locations for 
those modems can be found in $O(n \log n)$ time.
\end{thm}

\begin{proof}
 We can assume that the number of vertical segments in 
$\mathcal{L}$ is at least $n/2$, otherwise we can rotate the plane. Let $\mathcal{P}'$ 
be the set of 
endpoints of the segments in $\mathcal{L}$, excepting the 
lower endpoints of the vertical segments. Therefore, $|\mathcal{P}'| \le 3n/2$. 
\emph{StaircasePartition($\mathcal{L}$)} adds staircases to the partition 
only when it makes a call to \emph{OverflowCut}. The staircases added are at most two. Using Lemma 
\ref{lem:num-modems}, we obtain 
that there are at most $2|\mathcal{P}'|/ \left \lceil k/2 \right \rceil + 1 \le 6\frac{n}{k}+1$ 
staircases in the partition. Since each staircase is illuminable with one $k$-modem, the bound 
follows.

It remains to prove the $O(n \log n)$ bound on the time to find the locations of the $k$-modems. 

For each staircase 
$P$ in the partition,  we maintain the horizontal segments of $Steps(P)$ sorted  from 
left to right and the vertical segments of $Steps(P)$ sorted from bottom to top.
Thus it is possible to find $Above(P,l)$ and $Below(P,l)$ in $O(\log n)$ time; the same happens 
with 
$Left(P,l')$ and $Right(P,l')$. An \emph{OverflowCut} call makes at most two splittings, so it 
takes at most $O(\log n)$ time.
Since the number of \emph{OverflowCut} calls is $O(n/k)$, the total time required by them is 
$O(\frac{n}{k} \log n)$.

Given $m$ staircases $P_1, \ldots , P_m$ intersected by a segment $s$, the splits and merges done 
by a \emph{CrossingCut} can be achieved in $(m-2)O(\log n)$ time. The $m-2$ upper parts of the 
staircases splitted horizontally become staircases that won't be modified again. Since there 
are at most $O(n/k)$ staircases, at 
most $O(n/k)$ splits and merges from \emph{CrossingCut} are done throughout the algorithm. Since 
each pair of split and 
merge operations takes $O(\log n)$ time, the total 
time required by the \emph{CrossingCut} calls is $O(\frac{n}{k} \log n)$. 

Besides the cost of the calls to \emph{OverflowCut} and \emph{CrossingCut}, at every stop point, 
we must determine the staircase where it is located. This can be done in $O(\log n)$ time by 
maintaining the order from left to right in which the sweep line intersects the 
staircases. Thus, 
the running time for 
the algorithm is $O(n \log n + \frac{n}{k} \log n)$.
\end{proof}

\section{Aknowledgements}
We would like to thank Ruy Fabila-Monroy for introducing us to the $k$-modem problem 
and for his invaluable help in the development of this paper.

\bibliographystyle{plain}
\bibliography{kmodems}

\begin{thebibliography}{10}

\bibitem{Aggarwal:1984:AGT:911725}
A.~Aggarwal.
\newblock {\em The Art Gallery Theorem: Its Variations, Applications and
  Algorithmic Aspects}.
\newblock PhD thesis, 1984.
\newblock AAI8501615.

\bibitem{monotone}
O.~Aichholzer, R.~Fabila-Monroy, D.~Flores-Pe{\~n}aloza, T.~Hackl, C.~Huemer,
  J.~Urrutia, and B.~Vogtenhuber.
\newblock {{Modem Illumination of Monotone Polygons}}.
\newblock In {\em Proc. $25^{th}$ European Workshop on Computational Geometry
  EuroCG '09}, pages 167--170, Brussels, Belgium, 2009.

\bibitem{metaheuristic}
A.L. Bajuelos, S.~Canales, G.~Hern{\'a}ndez, and M.~Martins.
\newblock A hybrid metaheuristic strategy for covering with wireless devices.
\newblock {\em J.UCS}, 18(14):1906--1932, 2012.

\bibitem{ktransmiters}
B.~Ballinger, N.~Benbernou, P.~Bose, M.~Damian, E.~D. Demaine, V.~Dujmovi\'c,
  R.~Flatland, F.~Hurtado, J.~Iacono, A.~Lubiw, P.~Morin, V.~Sacrist\'an,
  D.~Souvaine, and R.~Uehara.
\newblock Coverage with k-transmitters in the presence of obstacles.
\newblock In Weili Wu and Ovidiu Daescu, editors, {\em Combinatorial
  Optimization and Applications}, volume 6509 of {\em Lecture Notes in Computer
  Science}, pages 1--15. Springer Berlin Heidelberg, 2010.

\bibitem{cuttingLemma1}
B.~Chazelle and J.~Friedman.
\newblock A deterministic view of random sampling and its use in geometry.
\newblock {\em Combinatorica}, 10(3):229--249, 1990.

\bibitem{berg1995cuttings}
M.~de~Berg and O.~Schwarzkopf.
\newblock Cuttings and applications.
\newblock {\em International Journal of Computational Geometry \&
  Applications}, 5(04):343--355, 1995.

\bibitem{fabila2009modem}
R.~Fabila-Monroy., A.R. Vargas, and J.~Urrutia.
\newblock On modem illumination problems.
\newblock In {\em XIII Encuentros de Geometria Computacional}, Zaragoza, Spain,
  2009.

\bibitem{CuttingLemma2}
J.~Matou\v{s}ek.
\newblock Construction of epsilon nets.
\newblock In {\em Proceedings of the Fifth Annual Symposium on Computational
  Geometry}, SCG '89, pages 1--10, New York, NY, USA, 1989. ACM.

\bibitem{o1987art}
J.~O'rourke.
\newblock {\em Art gallery theorems and algorithms}, volume~57.
\newblock Oxford University Press Oxford, 1987.

\bibitem{sack1999handbook}
J.R. Sack and J.~Urrutia.
\newblock {\em Handbook of computational geometry}.
\newblock Elsevier, 1999.

\end{thebibliography}

\end{document}